\documentclass[twoside,leqno,twocolumn]{article}

\usepackage[letterpaper]{geometry}

\usepackage{ltexpprt}

\usepackage{graphicx}
\usepackage{subcaption}
\usepackage{enumitem}
\usepackage{amsmath}
\usepackage{amsfonts}
\usepackage{dblfloatfix}
\usepackage{cite}
\usepackage{balance}

\usepackage{hyperref}

\usepackage{cleveref}

\usepackage{algpseudocode} 

\begin{document}

\title{Polycubes via Dual Loops}

{ 
    \author{Maxim Snoep\\{\small TU Eindhoven, The Netherlands}
    \and Bettina Speckmann\\{\small TU Eindhoven, The Netherlands}
    \and Kevin Verbeek\\{\small TU Eindhoven, The Netherlands}}    
}

\date{}

\fancyfoot[C]{\small\thepage}

\maketitle




\fancyfoot[R]{\scriptsize{Copyright \textcopyright\ 2025\\
Copyright for this paper is retained by the authors}}



\begin{abstract} \small\baselineskip=9pt In this paper we study polycubes: orthogonal polyhedra with axis-aligned quadrilateral faces. We present a complete characterization of polycubes of any genus based on their dual structure: a collection of oriented loops which run in each of the axis directions and capture polycubes via their intersection patterns. A polycube loop structure uniquely corresponds to a polycube. We also describe all combinatorially different ways to add a loop to a loop structure while maintaining its validity. Similarly, we show how to identify loops that can be removed from a polycube loop structure without invalidating it. Our characterization gives rise to an iterative algorithm to construct provably valid polycube maps for a given input surface.
\end{abstract}

\begin{figure}[b]
    \hfill
    \includegraphics[width=0.45\linewidth]{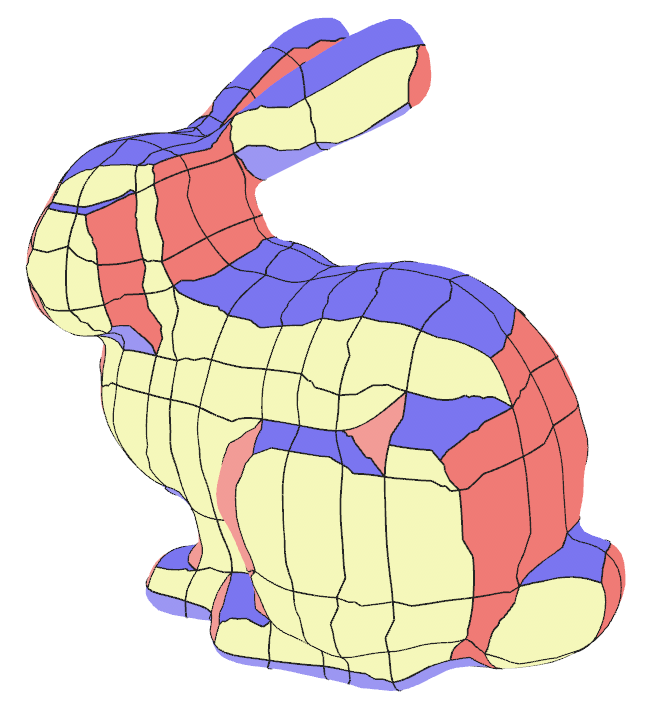}
    \hfill
    \includegraphics[width=0.45\linewidth]{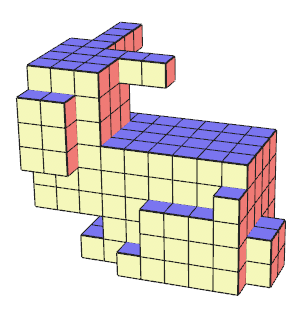}
    \hfill\quad
    \caption{Polycube segmentation of the Stanford bunny.}
    \label{fig:polycubesegmentation}
\end{figure}

\section{Introduction}
Polycubes are orthogonal polyhedra with axis-aligned quadrilateral faces. The simple structure of polycubes enables efficient solutions to various challenging geometric problems. 
Bijective mappings from general shapes to polycubes, known as \emph{polycube maps}, enable the transfer of solutions computed on polycubes to more general shapes. We distinguish between \emph{volumetric maps}, which map both the interior and the surface of the polycube, and \emph{surface maps} which map only the surface of the polycube to the target.
Polycube maps are used to solve problems such as texture mapping~\cite{tarini2004polycube}, spline fitting~\cite{wang2007polycube}, and hexahedral meshing~\cite{pietroni2022hex}. We focus on surface maps only; there are established methods~\cite{gregson2011all, protais2022robust,yu2014optimizing} that extend a surface map to a volumetric map.

Formally, a polycube map $f$ is a continuous map from a polycube $Q$ of genus $g$ to a closed $2$-dimensional surface $\mathcal{S}$ of (typically) the same genus. The edges of $Q$ map to a segmentation of $\mathcal{S}$ into \emph{patches} that correspond to the faces of $Q$, known as a \emph{polycube segmentation} (see Figure~\ref{fig:polycubesegmentation}). We distinguish two types of edges of~$Q$: \emph{creased edges} which separate two faces of different orientation and \emph{flat edges} which separate two coplanar faces.
In general, the quality of a polycube map is determined by two conflicting criteria: the complexity of~$Q$ and the distortion introduced by the mapping $f$.

Since polycubes were introduced in 2004, many methods have been proposed to construct polycube maps for a given input surface $\mathcal{S}$, see~\cite{dumery2022evocube,fang2016all,fu2016efficient,gregson2011all,guo2020cut,hu2016centroidal,hu2017surface,li2021interactive,livesu2013polycut,mandad2022intrinsic,yang2019computing,yu2014optimizing,zhao2017robust}.
Most of these methods~\cite{dumery2022evocube,fu2016efficient,gregson2011all,hu2016centroidal,hu2017surface,livesu2013polycut,zhao2017robust} work in the \emph{primal}: they attempt to directly create the polycube segmentation of $\mathcal{S}$ and derive the polycube and polycube map afterwards. The resulting polycubes are represented only by their creased edges and, correspondingly, their faces can be rectilinear polygons, instead of quadrilaterals only.
They use a variety of methods to construct a segmentation of $\mathcal{S}$ into surface patches and label each patch with an axis-direction 
$+X$, $-X$, $+Y$, $-Y$, $+Z$, or $-Z$. Such a \emph{labeled segmentation} is a valid polycube segmentation if and only if a corresponding polycube exists that preserves these labels. That is, there must exist a polycube $Q$ such that $(a)$ its faces correspond one-to-one to the surface patches of $\mathcal{S}$, and $(b)$ the normal vector of each face of $Q$ corresponds to the label of its matching surface patch.

Verifying whether a labeled segmentation is a valid polycube segmentation is a challenging problem. Most current primal methods rely on one of three characterizations for polycubes~\cite{eppstein2010steinitz, he2024expanding,zhao2019polycube}. Eppstein and Mumford characterize \emph{simple orthogonal polyhedra}~\cite{eppstein2010steinitz} which are genus-0 polyhedra with simply-connected faces and exactly three mutually perpendicular axis-aligned edges meeting at every vertex. These conditions are simple to check for a segmentation, but they do not cover polycubes of higher genus or vertex degrees higher than three (with respect to creased edges). 

\begin{figure}[b]
    \centering
    \hfill
    \subcaptionbox{}{\includegraphics[width=0.43\linewidth]{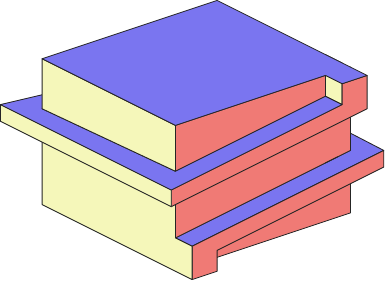}}
    \hfill
    \subcaptionbox{}{\includegraphics[width=0.43\linewidth]{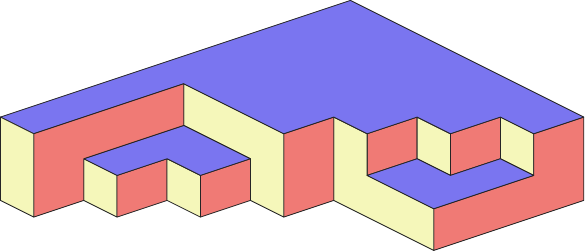}}
    \hfill\quad
    \caption{A segmentation and the polycube it (incorrectly) maps to following the characterization of \cite{eppstein2010steinitz}.}
    \label{fig:impossible_shape}
\end{figure}

The other two characterizations~\cite{he2024expanding, zhao2019polycube} attempt to lift these restrictions. However, none of the three characterizations takes the orientation of the surface patches into account. Consider the example in Figure~\ref{fig:impossible_shape}, first proposed by Sokolov~\cite{sokolov2016modelisation}. The segmentation on the left has four faces parallel to the $XZ$-plane (indicated in blue), two facing up and two facing down: the top and the bottom as well as the top and the bottom of the ``ramp''. According to all three characterizations, this segmentation corresponds to a polycube, namely the simple orthogonal polyhedron on the right. However, the polycube on the right clearly has three blue faces facing up and only one blue face facing down. This example shows that local conditions are not sufficient to characterize a polycube and that there are certain global constraints that must be satisfied. Methods that rely on a local characterization hence sometimes incorrectly classify labeled segmentations as valid. There is currently no method that is guaranteed to turn an invalid labeled segmentation into a valid one, let alone do so in an efficient manner.

An alternative set of primal methods~\cite{li2021interactive,yang2019computing,yu2014optimizing} project a voxelization (the polycube) of the input surface $\mathcal{S}$ back onto $\mathcal{S}$. By construction, the polycube segmentation is directly linked to a polycube. However, the mapping might create inverted faces within the segmentation. Furthermore, the voxelization may not preserve the topology of $\mathcal{S}$. In general, voxelization methods cannot guarantee the correctness of the polycube map.

For the specific case of hexahedral meshing, there are \emph{intrinsic} methods~\cite{fang2016all,guo2020cut,mandad2022intrinsic} which create volumetric polycube maps that are advantageously for downstream processing. The resulting hex meshes are generally of high quality, however, the distortion introduced by the polycube segmentation is often high and the topology of the input surface $\mathcal{S}$ may not be preserved.

Baumeister and Kobbelt~\cite{baumeister2023how} recently used a different approach for characterizing polycubes: they represent polycubes as quad meshes, including both flat and creased edges, and analyze their \emph{dual structure}. 
Since quad meshes have quadrilateral faces, their dual structure consists of a collection of loops which run in each of the axis directions and capture global properties by virtue of their intersection patterns. A version of this loop structure was described already 20 years ago by Biedl and Genc~\cite{biedl2004when}.
Both versions characterize only polycubes of genus 0.
Baumeister and Kobbelt check if the collection of dual loops implied by an input quad mesh corresponds to a valid polycube according to their characterization.

Verifying that a quad mesh, including flat and creased edges, corresponds to a polycube is computationally much simpler than verifying that a labeled surface segmentation corresponds to a polycube with arbitrarily complex rectilinear faces bounded by creased edges only. Note that rectilinear faces can first be partitioned into quadrilateral faces and then the verification based on dual loops can be employed. However, partitioning a rectilinear face in such a manner that the potential validity of the corresponding polycube is maintained, is an open problem.

Given an embedding of a valid dual loop structure on the input surface $\mathcal{S}$ one can derive the polycube segmentation through a simple primalization step~\cite{born2021layout, schreiner2004inter}.
In the remainder of the paper we hence focus on the characterization and construction of (quad mesh) polycubes of arbitrary genus via dual loops.
Just as previous work, we use sets of loops labeled with the three principal axes $X$, $Y$, and $Z$. However, our loops are oriented. This seemingly small change simplifies the characterization of polycubes of genus 0 and naturally extends to polycubes of higher genus.

We also describe all combinatorially different ways to add a loop to an existing loop structure while maintaining its validity. These valid loops correspond to cycles in a specific graph and can hence be detected and enumerated efficiently. Similarly we can detect all loops that can be removed from our loop structure without invalidating it.

Our characterization gives rise to an algorithm that constructs polycube segmentations: iteratively construct a valid dual structure on the input surface $\mathcal{S}$, starting from a simple polycube of the correct genus. 
At any point during the construction, there is a corresponding polycube segmentation. We add and subtract loops until the quality of the polycube segmentation is satisfactory.

Our paper is organized as follows. We define polycubes and describe their properties in Section~\ref{sec:definition}. We present our polycube characterization via labeled and oriented loops in  Section~\ref{sec:characterization} and show how to modify a loop structure while maintaining validity in Section~\ref{sec:modifications}. We implemented a proof-of-concept version of the iterative algorithm sketched above and showcase results in Section~\ref{sec:application}. Finally, we conclude in Section~\ref{sec:conclusion} by discussing several directions for future work.

\begin{figure*}[b]
    \centering
    \subcaptionbox{}{\includegraphics[]{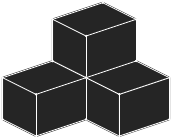}}
    \hfill
    \subcaptionbox{}{\includegraphics[]{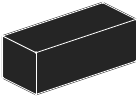}}
    \hfill
    \subcaptionbox{}{\includegraphics[]{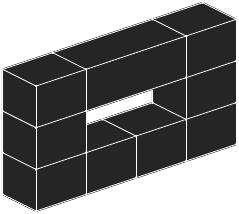}}
    \hfill
    \subcaptionbox{}{\includegraphics[]{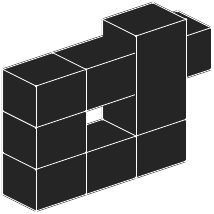}}
    \caption{The variety of polycubes that satisfy our definition based on quadrilateral meshes.}
    \label{fig:definition}
\end{figure*}

\section{Polycube definition and properties}\label{sec:definition}

The literature on polycube maps contains various definitions of polycubes with subtle differences. For our results, we use the fact that polycubes can be defined as a special type of quad mesh with axis-aligned faces~\cite{baumeister2023how}. In this case, we define the polycube as its $2$-dimensional boundary surface. We assume that the polycube has no voids, meaning the surface is connected and encloses a single bounded volume. As mentioned in the introduction, we focus on surface maps between the polycube and the input surface only.

A \emph{quadrilateral mesh (quad mesh)} consists of vertices, edges, and quadrilateral faces. Each vertex is adjacent to at least one edge. Each edge is adjacent to one or two faces. Each face consists of four vertices and four edges. A quad mesh is \emph{closed} if each edge is adjacent to exactly two faces. A quad mesh is \emph{orientable} if a consistent circular ordering of vertices can be assigned to each face, such that edge-adjacent faces have opposite vertex orders along their common edge. A quad mesh is \emph{connected} if every vertex can be reached from any other vertex by traversing edges. 

\begin{Definition}\label{def:polycube}
    A \emph{polycube} $Q$ is a closed, connected, orientable quad mesh with vertices $V(Q)$ such that:
    \begin{enumerate}[nolistsep]
        \item Each vertex $v \in V(Q)$ has a position $p(v)$ in $\mathbb{Z}^3$,
        \item Each vertex has degree at least $3$,
        \item Positions of adjacent vertices differ in exactly one coordinate,
        \item Edges incident to the same vertex cannot overlap.
    \end{enumerate}
\end{Definition}

As a consequence, our polycubes have vertices with degrees up to six (see Figure~\ref{fig:definition}a). 
The polycube faces are not required to be unit squares (see Figure~\ref{fig:definition}b), as this may not be general enough for higher genus polycubes (see Figure~\ref{fig:definition}c).
According to this definition, polycubes are also allowed to self-intersect, since this does not pose a problem for surface maps~\cite{sokolov2015fixing} (see Figure~\ref{fig:definition}d). Note that self-intersections might cause problems for volumetric methods, such as hex meshing.

Definition~\ref{def:polycube} extends the most common definition of polycubes (e.g.~\cite{livesu2013polycut}) in the sense that it allows a higher vertex degree. However, it is restricted to quadrilateral faces. Furthermore, most intrinsic methods (e.g.,\cite{mandad2022intrinsic}) allow for non-manifold vertices with overlapping edges, while our definition enforces ma\-ni\-fold vertices.

In the following, we establish some basic properties of polycubes based on Definition~\ref{def:polycube}.

\begin{lemma}\label{lem:rectface}
Every edge of a polycube is aligned with one of the coordinate axes ($X$-, $Y$-, or $Z$-axis), and every face of a polycube is an axis-aligned rectangle.
\end{lemma}
\begin{proof}
Any edge of the polycube connects two vertices, and by Condition 3 of Definition~\ref{def:polycube}, the positions of these vertices may differ in only one coordinate. Hence, every edge must be aligned with one coordinate axis. 

Now consider a quadrilateral face $F$ of a polycube. The four interior angles of $F$ must sum up to $360^\circ$. As all edges are axis-aligned, each vertex of $F$ must have an interior angle of $0^\circ$, $90^\circ$, $180^\circ$, or $270^\circ$. A vertex with a $0^\circ$ angle would violate Condition 4 of Definition~\ref{def:polycube}. Therefore, all vertices must have an interior angle of $90^\circ$, making $F$ a rectangle.\hfill
\end{proof}

Lemma~\ref{lem:rectface} together with Condition 4 of Definition~\ref{def:polycube} also directly imply the following property.

\begin{corollary} 
Each vertex in a polycube has at most six adjacent vertices. 
\end{corollary}

Each polycube defines three partial orders on its vertices, corresponding to the three principal axes ($X$, $Y$, and $Z$). The partial order for the $X$-axis is defined as follows: for two vertices $v$ and $w$, we say that $v \leq_X w$ if the $x$-coordinate of $v$ is less than or equal to the $x$-coordinate of $w$, and there is an edge between $v$ and $w$. The partial orders for the $Y$-axis and $Z$-axis are defined similarly. These partial orders play an important role in our dual characterization.

\bigskip

Definition~\ref{def:polycube} allows for different polycubes to have the exact same combinatorial structure, simply by changing the lengths of edges, see Figure~\ref{fig:order_equivalent}. We therefore use the partial orders to define a form of combinatorial equivalence between polycubes.

\begin{figure}[t]
    \centering

    \includegraphics[]{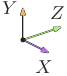}
    \hspace{0.1in}
    \includegraphics[scale=0.5]{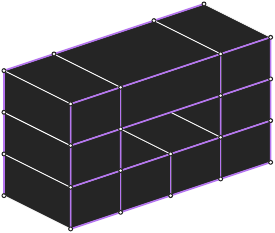}
    \hfill
    \includegraphics[scale=0.5]{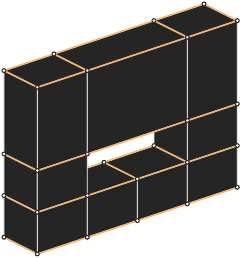}    
    \hfill
    \includegraphics[scale=0.5]{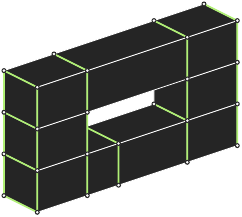}
    
    \caption{Order-equivalence of polycubes.}
    \label{fig:order_equivalent}
\end{figure}

\begin{Definition}
\label{def:order-equivalent}
    Two polycubes $Q_1$ and $Q_2$ are \emph{order-equivalent} if there exists an isomorphism $f\colon V(Q_1) \rightarrow V(Q_2)$ between the quad meshes of $Q_1$ and $Q_2$ such that, for all $v, w \in V(Q_1)$ and $\Delta \in \{X,Y,Z\}$, we have that $v \leq_\Delta w$ if and only if $f(v) \leq_\Delta f(w)$.
\end{Definition}

This definition implies that a polycube $Q$ is also order-equivalent to its inverse, which can be obtained by flipping the orientation of all faces of $Q$. However, the inverse does not enclose a bounded volume; it represents a single void.  
We represent the partial orders $\leq_X$, $\leq_Y$, and $\leq_Z$ as directed \emph{level graphs}: the $X$-graph, $Y$-graph, and $Z$-graph.
From now on we refer to the vertices of polycubes as \emph{corners} and reserve the term vertex for vertices in the level graphs and loop structures.

\section{Polycube dual characterization}
\label{sec:characterization}

Every embedded graph $G$ can be represented by its \emph{dual graph}, which has a vertex for each face of $G$ and an edge for each pair of adjacent faces of $G$. The dual of an (embedded) quad mesh is a $4$-regular graph (that is, every vertex has degree $4$),  since each face in the mesh is a quadrilateral. Any $r$-regular graph, with even $r$, can be decomposed into disjoint simple cycles where consecutive edges in each cycle do not share a face~\cite{petersen1891die}. 

We refer to an arrangement of simple closed curves $\mathcal{L}$ on a $2$-manifold as a \emph{loop structure}. 
Such loop structures can be represented as a graph, where \emph{loop intersections} define vertices, \emph{loop segments} (parts of a loop between two intersections) define edges, and \emph{loop regions} (regions bounded by loops) define faces.
Campen \emph{et al.}~\cite{campen2012dual} give the set of properties that characterize \emph{quad loop structures}: the loop structures which are a dual of quad meshes (see Figure~\ref{fig:quaddual}).

\begin{figure}[t]
    \centering
    \includegraphics[width=0.9\linewidth]{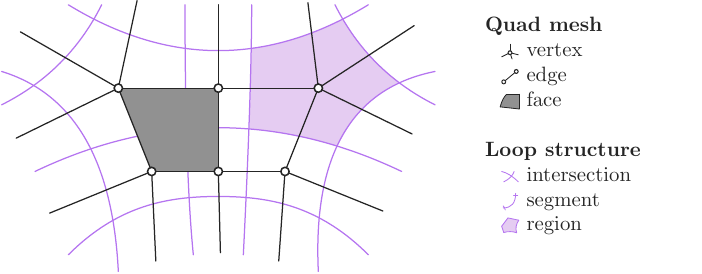}
    \caption{The duality of quad meshes and loops.}
    \label{fig:quaddual}
\end{figure}

\begin{figure}[b]
    {\centering
    \hfill
    \includegraphics[]{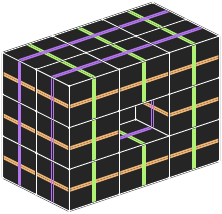}
    \hfill
    \includegraphics[]{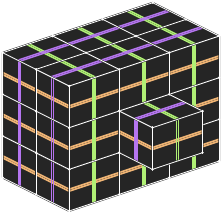}
    \hfill\quad
    }
    \caption{Two polycubes and their dual loop structure. Both polycubes have equivalent labeled loop structures.}
    \label{fig:equivalent}
\end{figure}

\begin{figure*}[b]
    \centering
    \hfill
    \includegraphics[scale=0.84]{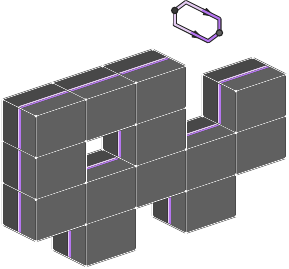}
    \hfill
    \includegraphics[scale=0.84]{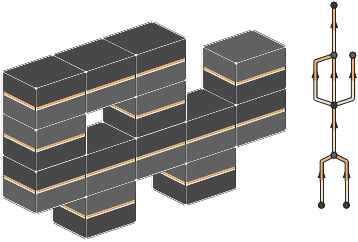}
    \hfill
    \includegraphics[scale=0.84]{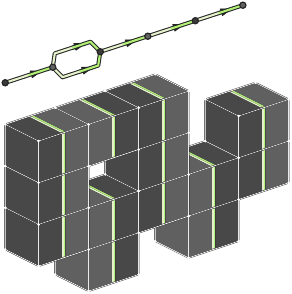}
    \hfill\quad
    \caption{Oriented loops and the corresponding level graphs. We use the colors purple, lighter purple, orange, lighter orange, green, and lighter green for $+X$, $-X$, $+Y$, $-Y$, $+Z$, $-Z$, respectively. For visual matching between the zones in the polycubes and in the level graphs, the zones are colored with two gray hues in alternating fashion.}
    \label{fig:zones_graph}
\end{figure*}

\begin{Definition}[Campen \emph{et al.}~\cite{campen2012dual}]
\label{def:quadstructure}
    A loop structure $\mathcal{L}$ is a \emph{quad loop structure} if:
    \begin{enumerate}[nolistsep]
        \item All loop intersections are transversal,
        \item No three loops intersect at a single point,
        \item Each loop region has the topology of a disk,
        \item Each loop region is bounded by at least two loop segments.
    \end{enumerate}
\end{Definition}
Since polycubes are quad meshes they do have a dual loop structure: each loop corresponds to a strip of quadrilateral faces whose center points share one coordinate~\cite{biedl2004when}, see Figure~\ref{fig:equivalent}. The loops can be classified as $X$-, $Y$-, or $Z$-loops. In our figures, we color the $X$-, $Y$-, and $Z$-loops with purple, orange, and green, respectively. An $X$-loop traverses faces whose normals align with the $Y$ and $Z$ axes, that is, the normals of the faces are perpendicular to the $X$ axis. Similarly, $Y$-loops and $Z$-loops traverse faces whose normals are perpendicular to their respective axes.

Baumeister and Kobbelt~\cite{baumeister2023how} explore how to transform a quad mesh into a polycube, by modifying its loop structure. Specifically, they label the loop structure of an input quad mesh with $X$, $Y$, or $Z$, such that the \emph{labeled loop structure} corresponds to the loop structure of a polycube. For the latter they define a characterization of the loop structure of polycubes of genus 0.
However, this characterization does not uniquely define a polycube. For example, the two polycubes in Figure~\ref{fig:equivalent} share the same labeled loop structure but differ in face alignment. The first polycube has a protruding block, while the second has an indented block. Additionally, this characterization applies only to genus-0 polycubes, and polycubes where any two dual loops intersect no more than twice. 

In the remainder of this section, we fix these limitations with the following minor change to the loop structure: we assign an \emph{orientation} (clockwise or counterclockwise) to every loop. The orientation of a loop can be interpreted as giving the loop two sides: a \emph{positive} side where the corresponding coordinate increases in the polycube, and a \emph{negative} side where the corresponding coordinate decreases. Therefore, to simplify our explanations, we orient the loops implicitly by assigning labels (positive or negative) to the two sides of a loop. This distinction between the loop sides is important for determining the direction of each polycube edge. In our figures containing oriented loops, we use two different shades to differentiate between the two sides of each loop. The negative ($-$) side is represented by a lighter shade compared to the positive ($+$) side. We refer to such an extended labeled loop structure as an \emph{oriented loop structure}, see Figure~\ref{fig:zones_graph}.

Before we can give a characterization of polycubes via oriented loop structures, we need to introduce a few concepts. We can use the full set of $X$-loops of an oriented loop structure to partition the underlying space (surface or polycube) into regions. We refer to these regions as \emph{$X$-zones}, or simply \emph{zones} in general. Similarly, we can use the $Y$-loops and $Z$-loops to obtain $Y$-zones and $Z$-zones, respectively. For each set of zones, we construct the $X$-, $Y$-, and $Z$-graph. The $X$-graph is constructed as follows (the $Y$- and $Z$-graphs are symmetric). We first add a vertex for each $X$-zone. Then, for each $X$-loop, we add a directed edge $(u,v)$, where $u$ is the zone on the $-$ side and $v$ is the zone on the $+$ side of the loop. We say that zones $u$ and $v$ share a loop. Since the underlying surface or polycube may have higher genus, parallel edges (where two zones share more than one loop) can occur, see \Cref{fig:zones_graph}. Observe that these graphs also define the partial order of Definition~\ref{def:order-equivalent}. Consider a directed edge $(u,v)$ in the $X$-graph. Both $u$ and $v$ correspond to distinct zones which contain vertices with a shared $x$-coordinate. Notice that the shared $x$-coordinate of the vertices in zone $u$ is strictly smaller than the shared $x$-coordinate of the vertices in zone $v$.

We can now give a general characterization for oriented loop structures that correspond to a polycube; we refer to such loop structures as \emph{polycube loop structures}.

\begin{Definition}
\label{def:dual}
    An oriented loop structure $\mathcal{L}$ is a \emph{polycube loop structure} if:
    \begin{enumerate}[nolistsep]
        \item No three loops intersect at a single point.
        \item Each loop region is bounded by at least three loop segments.
        \item Within each loop region boundary, no two loop segments have the same axis label \emph{and} side label.
        \item Each loop region has the topology of a disk.
        \item The level graphs are acyclic.
    \end{enumerate}
\end{Definition}
In the remainder of this section, we prove that, given a polycube loop structure $\mathcal{L}$, we can construct a single unique polycube $Q$ (up to order-equivalence), and conversely, that any arbitrary polycube has a corresponding dual polycube loop structure. We must first establish several key properties.

\begin{theorem}
\label{thm:nonintersecting}
    In a polycube loop structure $\mathcal{L}$, loops with the same axis label cannot intersect.
\end{theorem}
\begin{proof}
    Assume, for contradiction, that two loops $\lambda_1$ and $\lambda_2$ with the same axis label intersect in $\mathcal{L}$. Without loss of generality, we may assume that $\lambda_1$ and $\lambda_2$ are $X$-loops. The intersection $x$ of $\lambda_1$ and $\lambda_2$ is incident to four loop regions around $x$. Now consider the two loop regions around $x$ that are on the positive side of $\lambda_1$. These two loop regions are separated by $\lambda_2$, and hence one of these loop regions must also be on the positive side of $\lambda_2$. But this violates Condition 3 of Definition~\ref{def:dual}. 
    Thus, in a polycube loop structure, loops with the same axis label cannot intersect.\hfill 
\end{proof}

\begin{lemma}\label{lem:polycube_is_quad}
    A polycube loop structure is also a quad loop structure.
\end{lemma}
\begin{proof}
We only need to establish that Condition 1 of Definition~\ref{def:quadstructure} also holds for polycube loop structures. Assume, for the sake of contradiction, that an intersection between loops is not transversal, and let $\lambda$ be one of the loops. Then there must be a loop region incident to this intersection such that $\lambda$ occurs twice along its boundary (before and after the intersection). Clearly, both segments must have the same side label for that loop region, which violates Condition 3 of Definition~\ref{def:dual}. This is a contradiction, so every polycube loop structure is also a quad loop structure.\hfill
\end{proof}

\begin{theorem}
\label{thm:dual1}
    For every polycube $Q$ there exists a polycube loop structure $\mathcal{L}$ that forms the dual of $Q$. 
\end{theorem}
\begin{proof}
    First, we show how to construct an oriented loop structure $\mathcal{L}$ given a polycube $Q$. Since $Q$ is a quad mesh, we can obtain a quad loop structure from its dual graph. It remains to label and orient the loops. Consider any loop $\lambda$ in the quad loop structure. By definition, $\lambda$ must traverse opposite edges of each face of $Q$ it visits. Since every face of $Q$ is represented by an axis-aligned rectangle, opposite edges of a face must be parallel and also axis-aligned. Thus, all edges intersected by $\lambda$ must be parallel, and we can label $\lambda$ with the axis to which all intersected edges are parallel. Furthermore, by looking at the coordinates of the endpoints (corners) of the edges intersected by $\lambda$, we can assign $+$ to the side of $\lambda$ where the corner with the higher respective coordinate resides, and $-$ to the other side of $\lambda$. Observe that these sides must be consistent for all edges intersected by $\lambda$, since the faces are axis-aligned rectangles.

    We now verify that $\mathcal{L}$ satisfies all conditions of Definition~\ref{def:dual}:
    
    \begin{enumerate}[nolistsep]
        \item $\mathcal{L}$ is the dual of a quad mesh, and as such, no three loops intersect at a single point.

        \item Each corner of $Q$ has at least three adjacent corners. As each loop region corresponds to a corner of $Q$, each loop region boundary contains at least three distinct loop segments.

        \item Each corner of $Q$ has up to six incident edges, each with a distinct principal direction ($+X$, $-X$, $+Y$, $-Y$, $+Z$, $-Z$). Since each loop region corresponds to a corner of $Q$, and its boundary consists of loop segments corresponding to the edges incident on that corner, there can be at most one loop segment with each axis and side label combination. If two segments had the same label combination, this would imply the existence of two edges with the same direction at the corner, which is not possible.

        \item $\mathcal{L}$ is the dual of a quad mesh, and as such, each loop region has the topology of a disk.
        
        \item For the sake of contradiction, assume without loss of generality that the $X$-graph contains a cycle. Consider a vertex in that cycle, corresponding to a zone with a coordinate value $x$. Following a directed edge in the graph strictly increases the coordinate value. The existence of the cycle therefore implies that $x > x$, which is a contradiction.
    \end{enumerate}
\smallskip\noindent Thus, for every polycube $Q$ there exists a polycube loop structure $\mathcal{L}$ that forms the dual of $Q$.\hfill
\end{proof}

\begin{figure*}[b]
\centering
    \hfill
    \subcaptionbox{}{\includegraphics{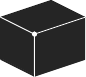}}
    \hfill
    \subcaptionbox{}{\includegraphics{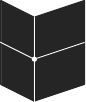}}
    \hfill
    \subcaptionbox{}{\includegraphics{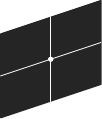}}
    \hfill
    \subcaptionbox{}{\includegraphics{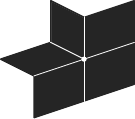}}
    \hfill
    \subcaptionbox{}{\includegraphics{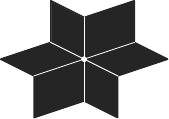}}
    \hfill
    \subcaptionbox{}{\includegraphics{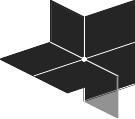}}
    \hfill\quad
    \caption{The six polycube corner categories.}
    \label{fig:vertex_categories}
\end{figure*}

\begin{theorem}
\label{thm:dual2}
    Given a polycube loop structure $\mathcal{L}$, there exists exactly one polycube $Q$ (up to order-equivalence) that corresponds to $\mathcal{L}$.
\end{theorem}
\begin{proof}
    We first show that a polycube $Q$ can be constructed from $\mathcal{L}$. Next, we show that any two polycubes $Q_1$ and $Q_2$ sharing the same polycube loop structure $\mathcal{L}$ must be the same polycube (up to order-equivalence).

    Following~\cite{campen2012dual}, a quad mesh can be constructed from $\mathcal{L}$, where each loop region corresponds to a corner in the mesh with degree equal to the number of loop segments bounding the loop region. We now assign a position in $\mathbb{Z}^3$ to each corner using the level graphs of $\mathcal{L}$. Since these graphs are acyclic, we can assign integer values to each of the vertices (zones) in the graphs such that they adhere to the strict partial order induced by the graph. Because every corner $v$ is in exactly one $X$-zone, one $Y$-zone, and one $Z$-zone, we can simply assign the position $p(v) = (x, y, z)$ to $v$, where $x$, $y$, and $z$ are the integers assigned to the $X$-zone, $Y$-zone, and $Z$-zone of $v$, respectively. Let $Q$ be the resulting embedded mesh. We now show that $Q$ is a polycube using the conditions of Definition~\ref{def:polycube}.
    
    \begin{enumerate}[nolistsep]
        \item Each corner has a position in $\mathbb{Z}^3$ by construction.

        \item Since each loop region is bounded by at least three distinct loop segments, each vertex must have at least three adjacent vertices.
        
        \item Neighboring loop regions share exactly two zones and differ in one zone, as they are separated by exactly one loop of $\mathcal{L}$. This results in adjacent vertices in the quad mesh that differ in exactly one coordinate, fulfilling this condition.
    
        \item For the sake of contradiction, assume that two edges $(v, u_1)$ and $(v, u_2)$ incident on the same corner $v$ overlap, and assume without loss of generality that they both point in the positive $X$-direction. Then both $u_1$ and $u_2$ must have been assigned a higher $x$-coordinate than $v$, and hence the $X$-graph must contain a directed edge from the zone of $v$ to the zone of $u_1$, and also from the zone of $v$ to the zone of $u_2$. But then the loop region of $\mathcal{L}$ corresponding to $v$ must have two loop segments on the boundary with label $X$ and side $-$. This contradicts Condition 3 of Definition~\ref{def:dual}.
    \end{enumerate}
    \smallskip
    Thus, $Q$ is indeed a polycube corresponding to $\mathcal{L}$. Now assume that two polycubes $Q_1$ and $Q_2$ correspond to $\mathcal{L}$. Then the quad meshes of $Q_1$ and $Q_2$ must be isomorphic, and the positions of the corners must adhere to the (strict) partial orders induced by the level graphs of $\mathcal{L}$. But then $Q_1$ and $Q_2$ must be order-equivalent by Definition~\ref{def:order-equivalent}.\hfill
\end{proof}
Another nice property of the $X$-, $Y$-, and $Z$-graphs of a polycube loop structure is that they capture the topological properties of the corresponding polycube, as demonstrated by the following theorem.

\begin{theorem}
    For any polycube loop structure $\mathcal{L}$, the $X$-, $Y$-, and $Z$-graphs each satisfy $|E| = g + |V| - 1 - \Gamma$, where $|E|$ and $|V|$ indicate the number of edges and vertices in a graph, $g$ is the genus of the corresponding polycube $Q$, and $\Gamma$ is the sum of genii for each zone.
\end{theorem}
\label{thm:genus}
\begin{proof}
    The Euler characteristic $\chi$ of a surface $\mathcal{S}$ with genus $g$ and $b$ boundaries is given by $\chi = 2-2g-b$.

    If $\mathcal{S}$ is decomposed into $k$ subsurfaces $\mathcal{S}_1, \ldots, \mathcal{S}_k$ by cutting along a collection of non-intersecting loops, where a subsurface $\mathcal{S}_i$ has Euler characteristic: $\chi_i = 2-2g_i-b_i$.
    
    Then we can obtain $\chi$ using the additive form:
    
    $$\chi = \sum^k_{i=1} \chi_i = \sum^k_{i=1} (2-2g_i-b_i)$$
    
    Let $\mathcal{S} = Q$ for some polycube $Q$ and consider the set of $X$-loops on $Q$. Let $l$ be the number of $X$-loops and let $k$ be the number of $X$-zones. Each loop is the boundary of two zones, thus there are a total of $2 l$ boundaries among the $k$ zones. We can simplify further to obtain:
    
    $$ \chi = 2k - 2l + \sum^k_{i=1} (-2g_i)$$
    
    Given that $\chi = 2 - 2g$, we derive the relation $2 - 2g = 2k - 2l +\sum^k_{i=1} (-2g_i)$, or equivalently, $l = g + k - 1 - \sum^k_{i=1} g_i$. The $X$-, $Y$-, and $Z$-graphs have $|V|=k$ and $|E|=l$. Recall $\Gamma$ is defined as $\sum^k_{i=1} g_i$. Therefore, we know that for the $X$-, $Y$-, and $Z$-graph each satisfy $|E|=g+|V|-1 - \Gamma$.\hfill
\end{proof}
Theorem~\ref{thm:genus} also directly implies the following nice property for polycubes of genus $0$.

\begin{corollary}\label{cor:tree}
    If $Q$ is a polycube of genus $0$, then the $X$-, $Y$-, and $Z$-graphs of its loop structure $\mathcal{L}$ are trees.
\end{corollary}

\begin{figure*}[b]
    \centering

    \vspace{0.15in}
    \includegraphics{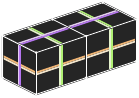}
    \hspace{0.5cm}
    \includegraphics{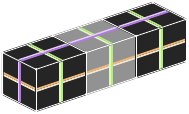}
    \hspace{0.5cm}
    \includegraphics{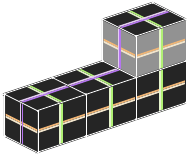}
    \hspace{0.5cm}
    \includegraphics{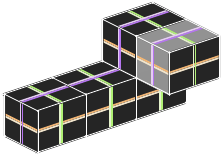}
    
    \caption{Sequence of loop additions: a $Z$-loop, then a $Y$-loop, and then a $X$-loop. Note that loops can be added between existing loops, as in the second figure.}
    \label{fig:addloop}
\end{figure*}

\smallskip\noindent{\bf Orienting a labeled loop structure.}
The introduction of oriented loop structures naturally raises the following question: can we always orient the loops of a labeled loop structure such that the result is a polycube loop structure? The work of Baumeister and Kobbelt~\cite{baumeister2023how} addresses this question indirectly, but only for polycubes of genus $0$. Here, we explicitly describe how to orient loops, or assign direction labels to their sides, to form a polycube loop structure for polycubes of arbitrary genus.

The conditions of Definition~\ref{def:dual} that are affected by the orientation of the loops are Conditions~3 and~5. Note that we can still compute the $X$-, $Y$-, $Z$-graphs of a labeled loop structure, but we cannot assign directions to the edges in the graphs. To assign orientations to the loops, we construct a graph on the loops for each axis separately ($X$, $Y$, or $Z$). For the $X$-axis, we make a graph that has two nodes for each $X$-loop, one corresponding to each side of the loop. We connect the two nodes of a loop with an edge, and we also connect two nodes if the corresponding loop sides both occur on the boundary of a single loop region. If this graph admits a $2$-coloring (that is, if it is bipartite) for all axes, then we can consistently assign side labels ($+$/$-$) to the sides of loops, or equivalently, orient the loops such that Condition~3 of Definition~\ref{def:dual} is always satisfied. Finally, we must also check Condition~5 of Definition~\ref{def:dual}: the $X$-, $Y$-, $Z$-graphs are acyclic. Note that, by Corollary~\ref{cor:tree}, this condition can never be violated for polycubes of genus $0$. For polycubes of higher genus, the challenge is to find a $2$-coloring of each of the $3$ graphs described above such that the $X$-, $Y$-, $Z$-graphs become acyclic. If the graph to be colored has $k$ connected components, then the graph has $2^k$ possible $2$-colorings. So, even though we can exhaustively check if a labeled loop structure can be oriented to become a polycube loop structure, it remains open if this can be computed efficiently.

\smallskip\noindent{\bf Polycube corners.}
For completeness we show all possible polycube corner configurations that occur and how they are represented in a polycube loop structure. At each polycube corner, the edges align with one of the coordinate axes ($X$, $Y$, or $Z$) and extend in either the positive or negative direction.

Each polycube corner must have at least three incident edges, and its incident edges do not include more than one edge of the same type. Additionally, adjacent edges in the arrangement must correspond to different labels, as loops of the same label do not intersect. These criteria are complete to enumerate all unique polycube corner configurations. In total, there are $16 + 24 + 6 + 48 + 8 + 24 = 126$ unique configurations. We show all unique configurations in \Cref{fig:all_vertices}. 

We can classify each configuration in six general patterns (see Figure~\ref{fig:vertex_categories}). The simple corner with only 3 adjacent edges; The edge corner and flat corner with 4 adjacent edges (one using all three axes, the other using only two axes); The bent corner with 5 adjacent edges; and the symmetric and asymmetric complex corner (one with a symmetric intersection pattern, the other with an asymmetric intersection pattern).

\section{Polycube modifications}
\label{sec:modifications}

To construct a polycube map or polycube segmentation for an input surface $\mathcal{S}$, we have to choose a suitable polycube $Q$ to map from. Current methods for constructing polycube segmentations do not have a guarantee on the existence of a corresponding polycube $Q$ while they are constructing the segmentation. This creates friction between optimizing the quality of the polycube segmentation (complexity and alignment) and guaranteeing the validity of the polycube segmentation (that it actually corresponds to a valid polycube).

Instead, we suggest to construct the polycube segmentation iteratively, where we guarantee the validity of the segmentation during every step. Doing this directly using the polycube segmentation on $\mathcal{S}$ is challenging, but our dual characterization of polycubes offers a relatively straightforward approach: we initialize a simple polycube loop structure $\mathcal{L}$ (with the same genus as $\mathcal{S}$) on the surface $\mathcal{S}$, and then we iteratively add or remove oriented and labeled loops to $\mathcal{L}$ whilst ensuring that it remains a polycube loop structure. We can then also easily keep track of the corresponding polycube $Q$ using Theorem~\ref{thm:dual2}. The loops we add or remove can be chosen based on some quality criteria for the final polycube segmentation. Although this approach yields only a polycube loop structure $\mathcal{L}$ (on $\mathcal{S}$), this loop structure prescribes the global structure of the polycube segmentation and only local optimizations are needed to obtain the final polycube segmentation.

We can simply check the characterization after each step to determine if an addition or removal is valid. Doing so is computationally expensive, especially when considering additions, since there are many potential candidate loops. Therefore we show in the following how to construct valid loops for addition, and to determine valid loops for removal, efficiently.  

\subsection{Loop addition.}
\label{sec:loopaddition}
The addition of a loop to a polycube loop structure $\mathcal{L}$ introduces a strip of faces in the corresponding polycube $Q$, see Figure~\ref{fig:addloop}. We say that a loop is \emph{valid} if it can be added to $\mathcal{L}$ such that $\mathcal{L}$ remains a polycube loop structure. 

To effectively find valid loops for $\mathcal{L}$, we develop a directed graph structure that enables the enumeration of all possible valid loops.
We say that any two loops are \emph{combinatorially equivalent} if their intersection patterns (the loop segments they intersect) in $\mathcal{L}$ are equivalent. As such, we observe that all possible (valid or invalid) loops that could be added to a loop structure are characterized by the edge-graph of $\mathcal{L}$, which is defined as follows: there is a vertex for each loop segment in $\mathcal{L}$, and there is a directed edge from vertex $u$ to vertex $v$ if the two corresponding loop segments bound the same loop region in $\mathcal{L}$. We call this directed graph $G_L$. Note that this graph is directed since we care about the orientation of the loops. For an oriented loop, we assume that the right side (following its direction) is the positive ($+$) side, and the left side is the negative ($-$) side.

The validity of a loop in $G_L$ also depends on its label. Therefore, we will construct three different graphs, $G_V^X$, $G_V^Y$, and $G_V^Z$, one for each possible label, where each graph is obtained by removing directed edges from $G_L$. The idea is that a loop with label $\Delta \in \{X,Y,Z\}$ is valid if and only if it corresponds to a simple (oriented) cycle in $G_V^\Delta$. A cycle is simple if it visits every loop region of $\mathcal{L}$ at most once.

\bigskip\noindent{\bf Local constraints.}
When we add a loop $\lambda$ to $\mathcal{L}$, all loop regions that are visited will be split into two loop regions. All other loop regions are unaffected. Note that the splitting of a loop region, results in two new loop regions with the topology of a disk, so Condition~4 of Definition~\ref{def:dual} is never violated when adding a loop. 

The other local conditions of Definition~\ref{def:dual} (1-3) can be guaranteed by constructing the correct graphs $G_V^X$, $G_V^Y$, and $G_V^Z$. First of all, Condition 1 can always be guaranteed by construction, as there always exists a combinatorially equivalent loop with that property. For Conditions~2 and~3, observe that every (directed) edge in $G_L$ corresponds to an individual split of a loop region into two loop regions. For each label $\Delta \in \{X,Y,Z\}$ of the loop we can simply directly check, for every directed edge in $G_L$, if the corresponding split results in valid loop regions according to Conditions~2 and~3. We include a directed edge in $G_V^\Delta$ if and only if the corresponding split is valid. There are three cases to consider based on the number of $\Delta$-loop segments within the loop region. If the loop region contains two $\Delta$-loop segments, the edge must separate the two existing loops. If the loop region contains one $\Delta$-loop segment, the edge must not intersect the existing $\Delta$-loop segment. If the loop region contains no $\Delta$-loop segments, there are no conflicts, and all edges are permitted. There is one consistent orientation of the edges if the loop region contains one or two $\Delta$-loop segments, such that the new edge does not conflict with existing edges. Otherwise, edges can be directed in either direction. These three general cases are illustrated in Figure~\ref{fig:localconstraints3}, where all valid outgoing edges for a single loop segment (vertex in graph $G_V^X$) are shown.

Note that we need to require that every loop region is visited at most once; otherwise, there will be a new loop region where the same loop (the added loop) occurs twice on the boundary. Since the underlying surface $\mathcal{S}$ is orientable, the side label of the two loop segments in the direction of that loop region must also be the same, which would violate Condition~3 of Definition~\ref{def:dual}. Thus, a loop with label $\Delta \in \{X,Y,Z\}$ satisfies Conditions~1-4 of Definition~\ref{def:dual} if and only if it corresponds to a directed cycle in $G_V^\Delta$ that visits every loop region at most once (simple).

\begin{figure}[t]
    \centering

    \includegraphics{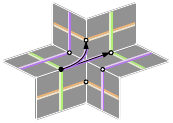}
    \hspace{0.0cm}
    \includegraphics{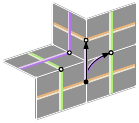}
    \hspace{0.3cm}
    \includegraphics{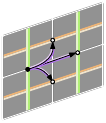}
    \hspace{0.1cm}
    
    \caption{The three cases for edges in $G_V^X$.}
    \label{fig:localconstraints3}
\end{figure}

\begin{figure*}[t]
    \centering

    \hfill
    \subcaptionbox{}{\includegraphics{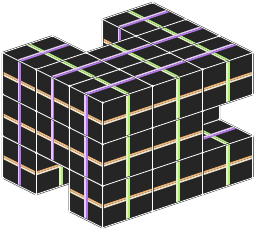}}
    \hfill
    \subcaptionbox{}{\includegraphics{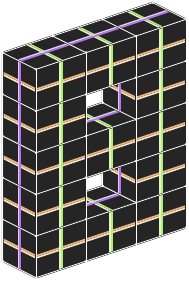}}
    \hfill
    \subcaptionbox{}{\includegraphics{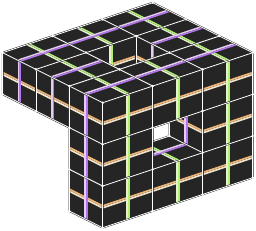}
    }
    \hfill\quad
    
    \caption{The non-incremental polycube that cannot be constructed from a single cube by loop additions~(a), and two higher-genus polycubes that cannot be transformed into one another by additions or removals~(b, c).}
    \label{fig:nonincremental}
\end{figure*}

\newpage

\smallskip\noindent{\bf Maintaining acyclicity.}
The only condition that remains to be checked is Condition~5 of Definition~\ref{def:dual}. We show that adding a loop can never introduce a cycle in the level graphs. Hence, a loop with label $\Delta \in \{X,Y,Z\}$ is valid if and only if it corresponds to a non-trivial simple directed cycle in $G_V^\Delta$.

Without loss of generality, consider the addition of an $X$-loop $\lambda$ to the polycube loop structure $\mathcal{L}$. Note that $\lambda$ cannot intersect another $X$-loop of $\mathcal{L}$. Thus, it will split an existing $X$-zone into two. By definition this zone is bounded by at least one loop. In the $X$-graph, the zone corresponds to a vertex $v$, and its boundary loops correspond to incoming or outgoing edges. The addition of $\lambda$ splits $v$ into two vertices $u$ and $w$ with a directed edge from $u$ to $w$. The incident edges of $v$ are distributed among $u$ and $w$. Observe that removing $\lambda$ from the resulting loop structure corresponds to contracting the vertices $u$ and $w$ into $v$. Thus, if there is a cycle in the $X$-graph that goes through $u$ or $w$ after adding $\lambda$, then there must have been a cycle through $v$ before adding $\lambda$, which contradicts the fact that $\mathcal{L}$ is a polycube loop structure. Finally note that the addition of an $X$-loop does not affect the $Y$- and $Z$-graphs. 

\subsection{Loop removal.}
When we consider removing a loop $\lambda$ from the polycube loop structure $\mathcal{L}$, we want to check efficiently if $\mathcal{L}$ without $\lambda$ is also a polycube loop structure. Without loss of generality, we assume that $\lambda$ is an $X$-loop. 

First of all, Condition~1 of Definition~\ref{def:quadstructure} cannot be violated by the removal of a loop. For the other local Conditions~2-4, observe that, when removing a loop $\lambda$ from $\mathcal{L}$, every loop segment of $\lambda$ causes two loop regions to be merged. We can easily find these loop regions by tracing $\lambda$ through $\mathcal{L}$ and then we can simply check if the resulting loop regions satisfy Conditions~2-4.

Finally, we need to check Condition~5 of Definition~\ref{def:dual} after removing $\lambda$ from $\mathcal{L}$. Note that $\lambda$ corresponds to an edge in the $X$-graph. Removing $\lambda$ does not only remove this edge, but also contracts the edge (since the zones corresponding to the endpoints are merged) into a new vertex $v$. This may introduce a cycle in the $X$-graph through $v$, but this can efficiently be checked using a simple depth first search from $v$. 

\begin{figure*}[b]
    \centering

    \includegraphics[width=\textwidth]{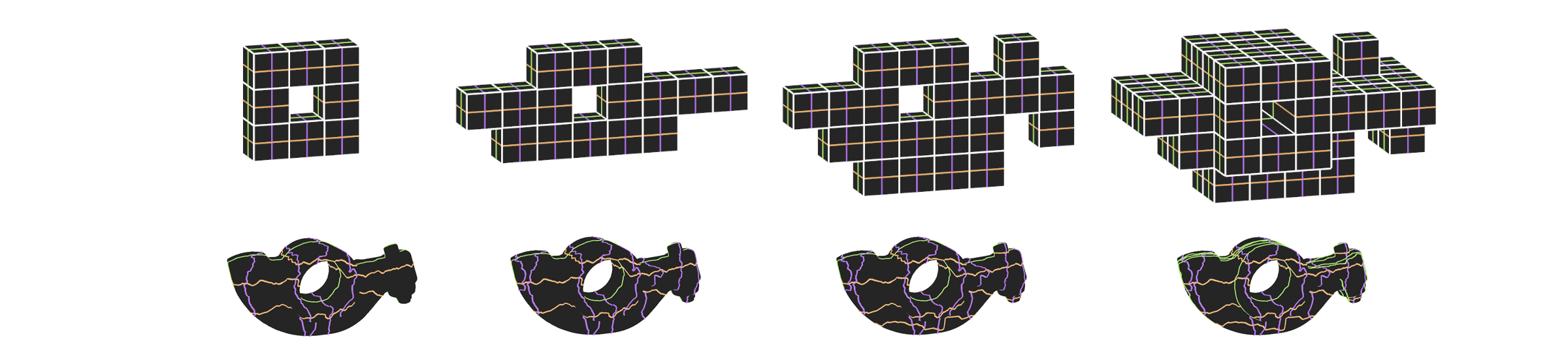}
    
    \caption{Iterative construction of a polycube map via its dual loop structure.}
    \label{fig:motion}
\end{figure*}

\subsection{Polycube reachability.}\label{sec:reachability}
Our iterative polycube modifications via adding or removing loops raise the following natural question: Can every polycube be constructed by adding and removing loops, starting from a polycube loop structure for the most simple polycube with the same genus?

We first consider polycubes of genus $0$, where the initial polycube loop structure corresponds to a single cube. In Figure~\ref{fig:nonincremental}(a), we show a polycube that cannot be obtained from the single cube by only adding loops. We refer to this polycube as the \emph{non-incremental polycube}. It is easily verified that none of the loops of the non-incremental polycube are valid loops for removal. It follows that it is impossible to construct the non-incremental polycube by iteratively adding valid loops. However, we can construct the non-incremental polycube by both adding and removing loops. In fact, we do not know if all polycubes of genus $0$ can be obtained by adding and removing loops, starting from the single cube. We leave this as an open problem.

We now move to polycubes of higher genus. Adding or removing loops cannot change the orientation of a hole in the initial polycube (see Figure~\ref{fig:nonincremental}(b, c)).

\section{Polycube segmentation construction}\label{sec:application}

Our characterization of polycubes gives rise to an iterative algorithm to construct a polycube segmentation of a given input surface $\mathcal{S}$ and, by extension, a polycube map and a hexahedral mesh. In the following, we first briefly sketch the steps that our proof-of-concept algorithm and its implementation take and then showcase our results. Note that neither algorithm nor implementation have been optimized for speed or quality; they simply serve to establish the feasibility of our approach.

We initialize our iterative algorithm with a single cube for genus-0 surfaces or a suitable polycube of genus $g$ for surfaces of genus $g$. In the future, we plan to develop an algorithm that automatically constructs suitable initial polycubes and corresponding loop structures for an input surface $\mathcal{S}$; for our proof-of-concept, we used manually-crafted polycubes and loop structures. To embed the initial loop structure on the input surface $\mathcal{S}$, we follow Campen and Kobbelt~\cite{campen2014dual} and use shortest paths in a suitably weighted graph on $\mathcal{S}$. Here, we can restrict ourselves to those parts of the graph that intersect specific loop segments (see Section~\ref{sec:loopaddition}). We attempt to embed loops to align with their label and orientation.

We now iteratively add or remove loops. Figure~\ref{fig:motion} illustrates the steps of our algorithm through a series of snapshots. It begins with an initial polycube dual loop structure for a genus-1 model, followed by successive steps adding X-loops, Y-loops, and Z-loops. To evaluate the quality of a loop structure we need to primalize it into a polycube segmentation. We first place a polycube corner into each loop region while attempting to align corners within the same zone, to create segmentation patches that are mostly rectangular. 
Then, we find non-intersecting paths on $\mathcal{S}$ to connect adjacent polycube corners, guided by the dual loops embedded on~$\mathcal{S}$. 

There are often many different loops that can be added to or removed from a valid loop structure. We use a simple evolutionary algorithm to choose loops that improve the  polycube segmentation. This algorithm uses mutation and selection only (no crossover); every candidate solution is a valid loop structure. To asses the quality of a solution, we use a metric that approximates the distortion of the corresponding polycube map: the alignment of each surface patch with its assigned axis. Figure~\ref{fig:showcase} shows some of our results for various models with varying complexity and genus.

\section{Future Work}\label{sec:conclusion}

Our proof-of-concept algorithm demonstrates the potential of our method.
While the current implementation is minimal and lacks extensive experimental evaluation, it nevertheless provides a glimpse of the future possibilities. 
Various avenues for future work remain.

On the theoretical side, we plan to establish which operations are necessary to construct any loop structure. We are currently not aware of a polycube of genus 0 that cannot be constructed via loop addition and removal. For polycubes of genus $g>1$, we also do not seem to require other operations while moving within the set of polycubes of genus $g$. However, loop addition and removal cannot re-orient holes or change the genus. 
Our current method constructs polycube maps which only map the surface of the polycube to the input surface. It would be interesting to see if and how our method extends to volumetric maps.

\begin{figure}[ht!]
    \centering

    \includegraphics[width=0.3\linewidth]{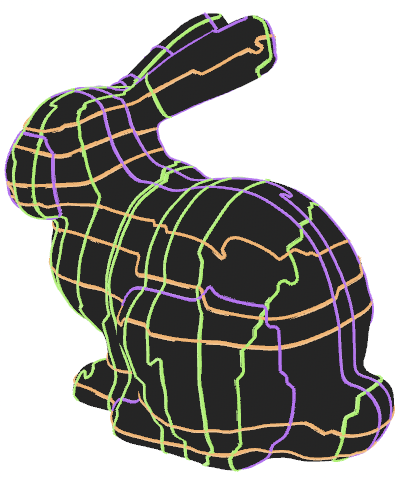}
    \includegraphics[width=0.3\linewidth]{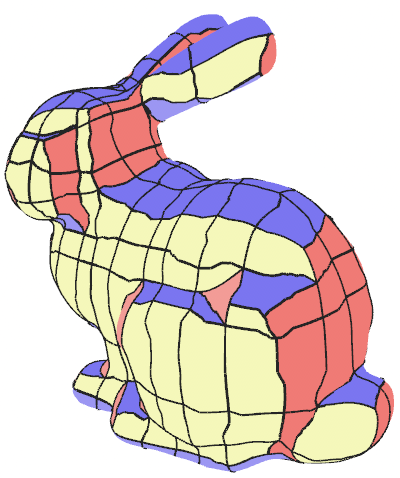}
    \includegraphics[width=0.3\linewidth]{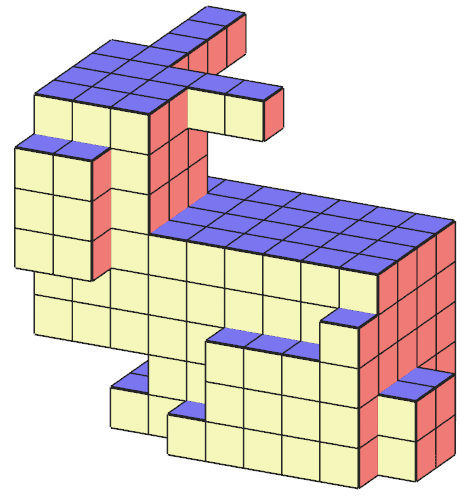}

    \vspace{0.1in}
    
    \includegraphics[width=0.3\linewidth]{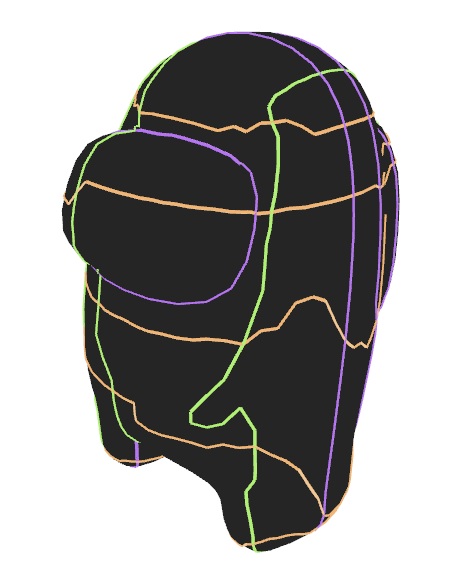}
    \includegraphics[width=0.3\linewidth]{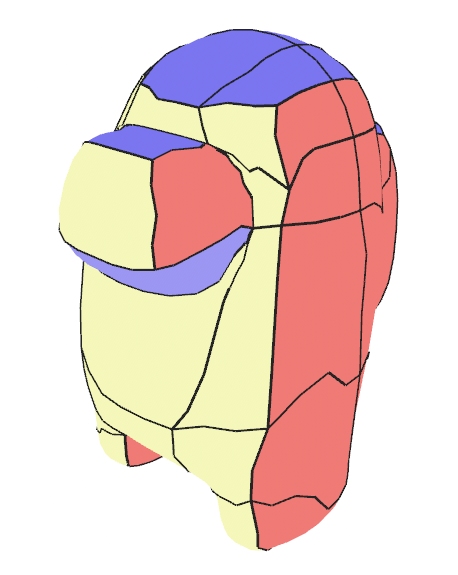}
    \includegraphics[width=0.3\linewidth]{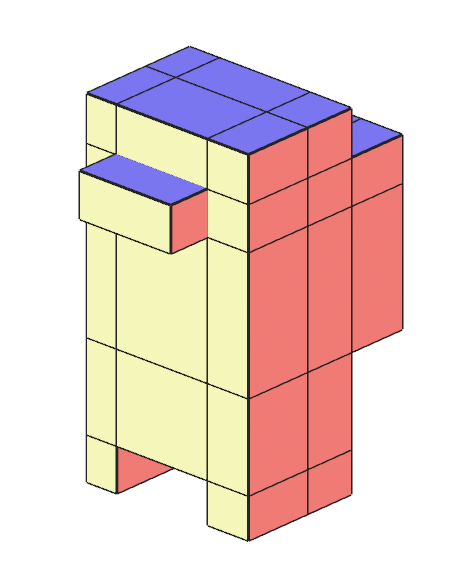}

    \vspace{0.1in}

    \includegraphics[width=0.3\linewidth]{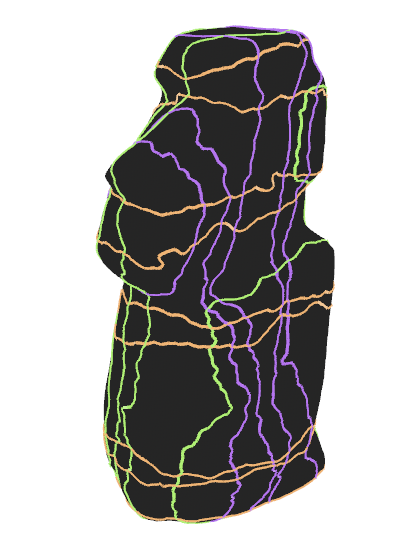}
    \includegraphics[width=0.3\linewidth]{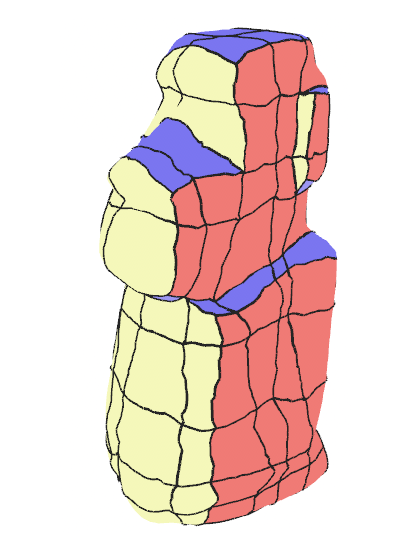}
    \includegraphics[width=0.3\linewidth]{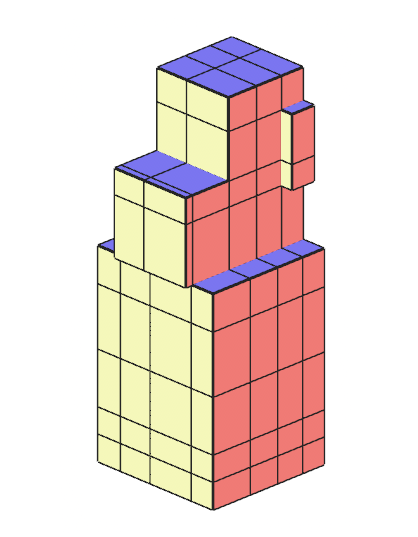}

    \vspace{0.05in}
    
    \includegraphics[width=0.3\linewidth]{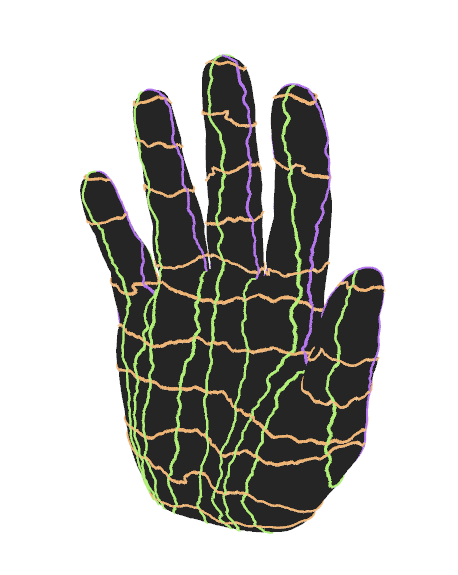}
    \includegraphics[width=0.3\linewidth]{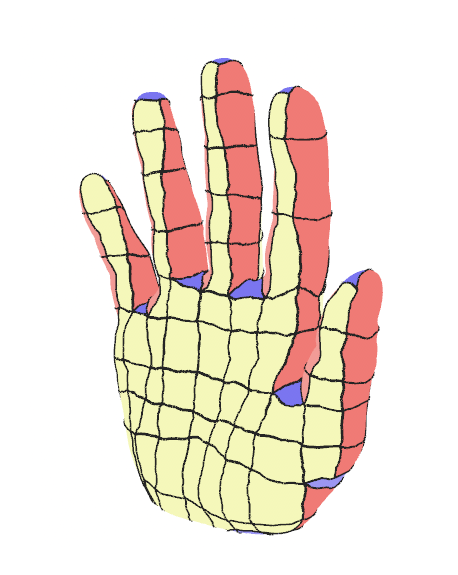}
    \includegraphics[width=0.3\linewidth]{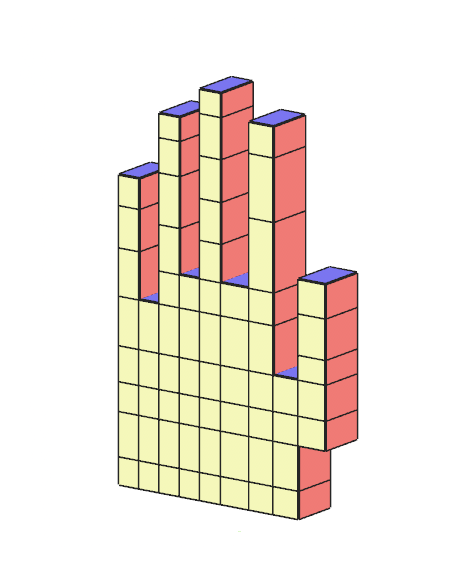}

    \vspace{0.05in}

    \includegraphics[width=0.3\linewidth]{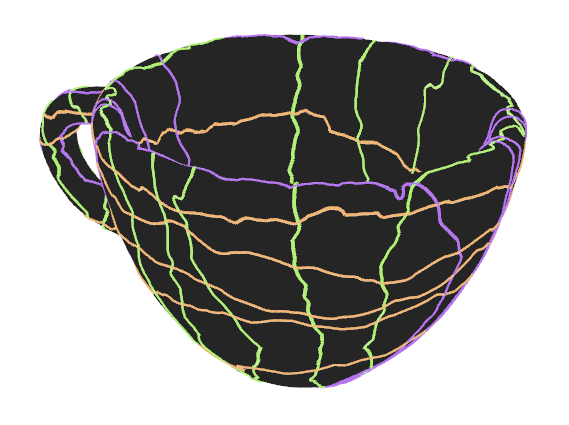}
    \includegraphics[width=0.3\linewidth]{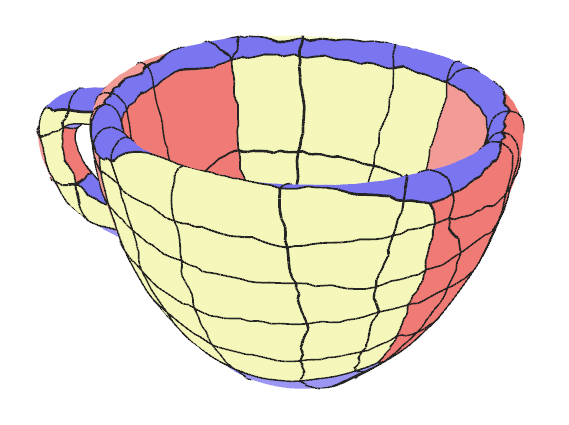}
    \includegraphics[width=0.3\linewidth]{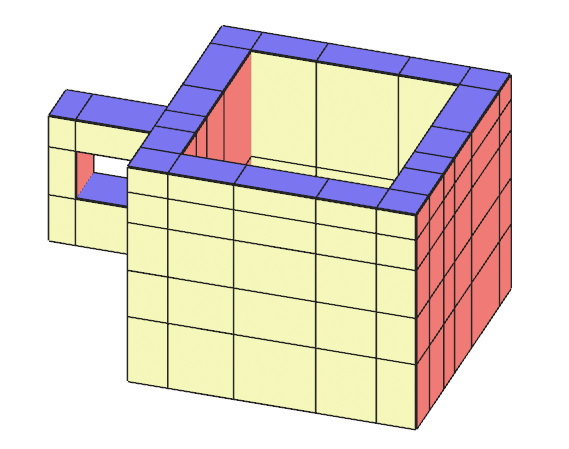}

    \vspace{0.02in}

    \includegraphics[width=0.3\linewidth]{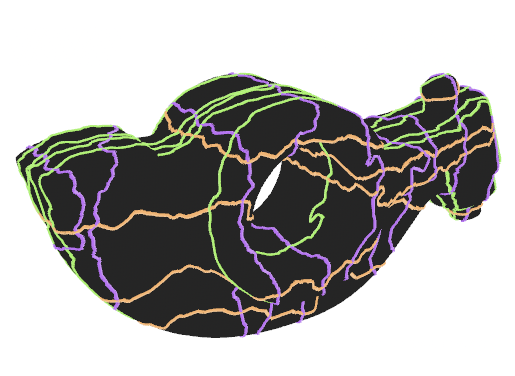}
    \includegraphics[width=0.3\linewidth]{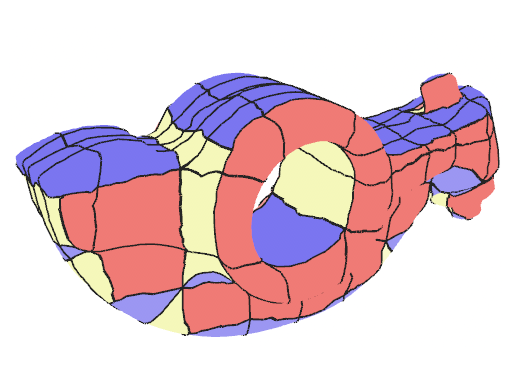}
    \includegraphics[width=0.3\linewidth]{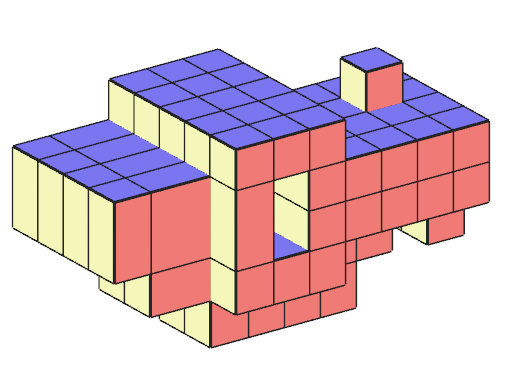}

    \vspace{0.1in}
    
    \includegraphics[width=0.3\linewidth]{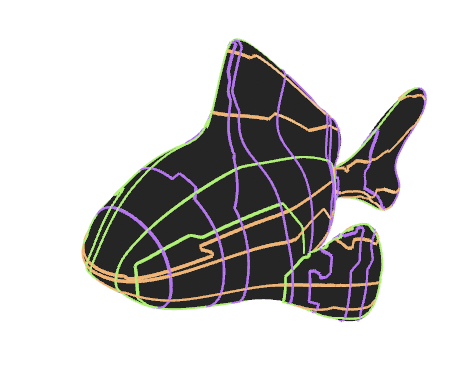}
    \includegraphics[width=0.3\linewidth]{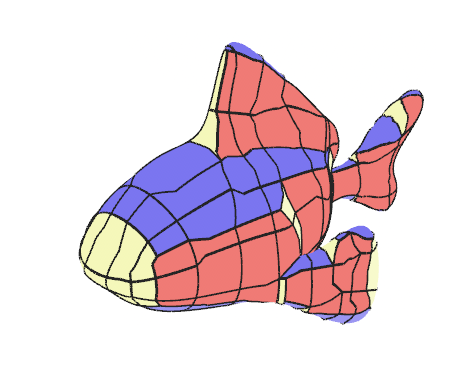}
    \includegraphics[width=0.3\linewidth]{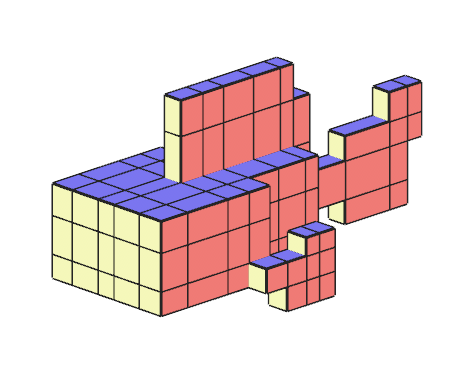}
    
    \caption{Embedded polycube loop structure, the polycube segmentation, and the corresponding polycube constructed by our proof-of-concept implementation.
    }
    \label{fig:showcase}
\end{figure}

On the practical side, we would like to develop an algorithm that constructs an initial polycube that captures global features of input surfaces of any genus. Furthermore, an optimal embedding of the loop structure on the surface should ideally align with salient features, such as ridges, valleys or handles. Finally, we would like to develop a method that finds an optimal primalization of our embedded loop structures, where optimality will necessarily depend on the use case. 

\begin{figure*}
    \centering
    \includegraphics[width=0.87\linewidth]{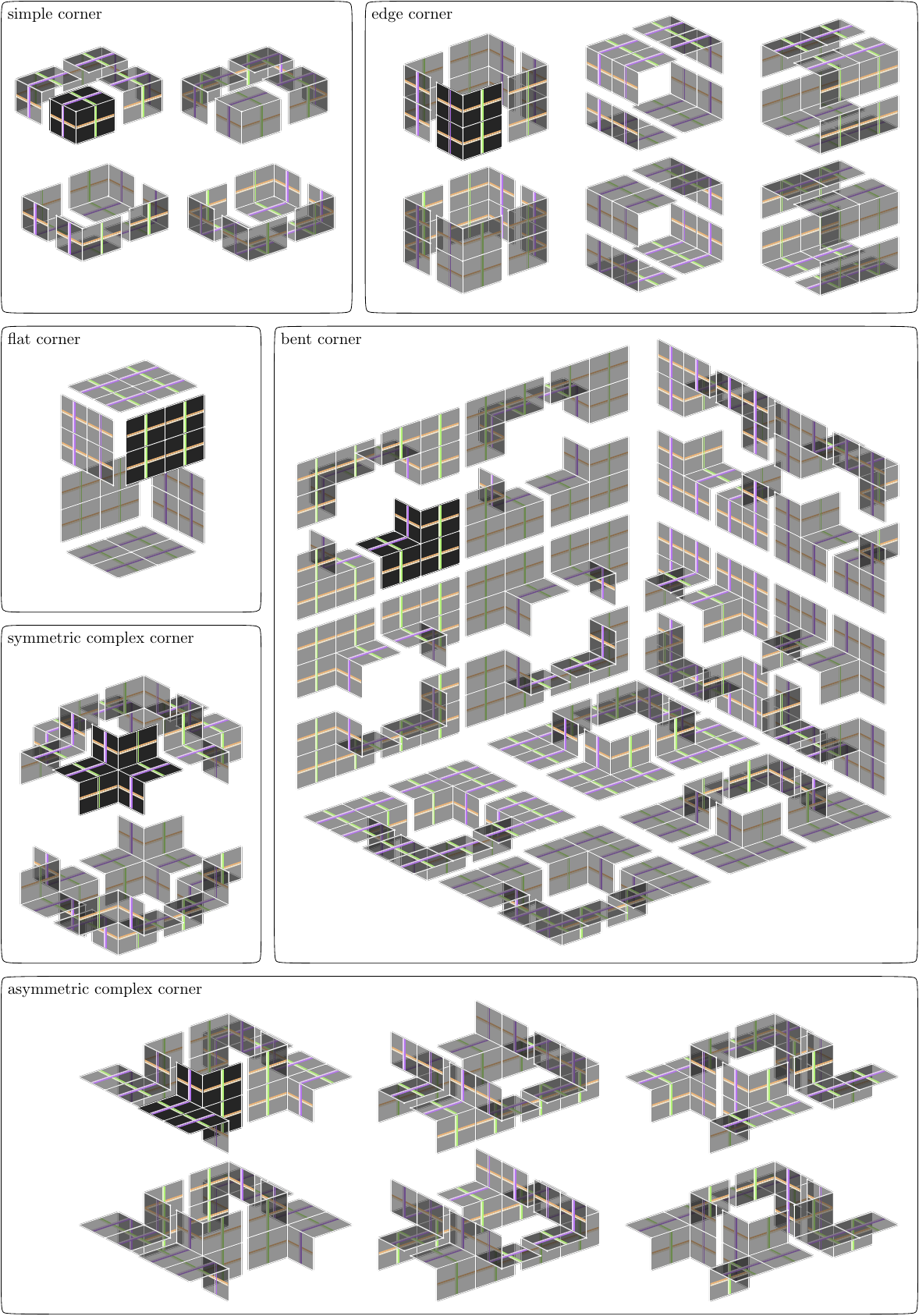}
    \caption{All unique polycube corners.}
    \label{fig:all_vertices}
\end{figure*}

\newpage

\balance

\bibliographystyle{siam}
\bibliography{bib}

\end{document}